\documentclass[journal, 10pt, twoside]{IEEEtran}

\usepackage{ifthen}
\usepackage[english]{babel}
\usepackage[fleqn]{amsmath}
\usepackage{amsfonts, amssymb, yhmath, bm}
\allowdisplaybreaks
\usepackage{graphicx}
\usepackage{color}
\usepackage[utf8]{inputenc}

\usepackage{psfrag}
\usepackage{todonotes}
\usepackage{microtype}

\usepackage{enumitem}
\newcommand{\mylistbullet}{
	\hfill \begin{tikzpicture}[scale=0.40]
		\clip (-.35,-.35) rectangle (.35, .35);
		\node[rounded corners=0.8pt, draw=black!40, fill=black!40]{$ $};
	\end{tikzpicture}
}
\newcommand{\enum}[1]{
	\begin{enumerate}[label=\arabic*),itemsep=1mm,leftmargin=0.6cm]
		#1
    \end{enumerate}
	
}
\newcommand{\enumrom}[1]{
	
	\begin{enumerate}[label=\roman*),itemsep=1mm,leftmargin=0.6cm]
		#1
    \end{enumerate}
	
}
\newcommand{\items}[1]{
	
	\begin{itemize}[label=\tiny\mylistbullet,itemsep=1mm,leftmargin=0.44cm]
		#1
    \end{itemize}
	
}

\newcommand{\Eq}[2]{\begin{equation}\label{eq:#2}
	#1
\end{equation}}

\usepackage[hidelinks,backref=page]{hyperref}
\usepackage[]{backref}
\renewcommand*{\backref}[1]{}
\renewcommand*{\backrefalt}[4]{%
{\ifcase #1 Not cited explicitly.%
\or cited once on page #2.%
\else cited #1 times on pages #2.%
\fi}}

\newcommand{\trans}{\mathrm{T}} % Vorsicht, noch nicht konsistent überall verwendet
\newcommand{\herm}{\mathrm{H}}

\newcommand{\R}{\bm{\mathrm{R}}}
\newcommand{\C}{\bm{\mathrm{C}}}

\newcommand{\N}{\bm{\mathrm{N}}}

\newcommand{\ScPr}[2]{{\left\langle #1,#2 \right\rangle}}
\newcommand{\Norm}[1]{{\left\Vert #1\right\Vert}}
\newcommand{\Abs}[1]{{\left| #1 \right|}}
\newcommand{\Text}[1]{{\hspace{3mm} \text{#1} \hspace{3mm}}}

\usepackage{calc}

\newcommand{\Sum}[3]{{\sum\limits_{#1}^{#2}{#3}}}
\newcommand{\Prod}[3]{{\prod\limits_{#1}^{#2}{#3}}}

\newcommand{\D}[0]{\hspace{0.5mm}:\hspace{1.0mm}}

\DeclareMathOperator*{\Diag}{diag}

\DeclareMathOperator*{\Max}{max}
\DeclareMathOperator*{\Min}{min}

\DeclareMathOperator*{\Rk}{rk}

\DeclareMathOperator*{\Supp}{supp}

\DeclareMathOperator*{\Vectorize}{vec}
\DeclareMathOperator*{\Unif}{Unif}

\usepackage{amsthm}
\usepackage{cleveref}

\newtheorem{Lem}{Lemma}[section]
\crefname{Lem}{lemma}{lemmata}

\newtheorem{Thm}{Theorem}[section]
\crefname{Thm}{theorem}{theorems}

\crefname{Pro}{proposition}{propositions}

\newtheorem{Cor}{Corollary}[section]
\crefname{Cor}{corollary}{corollaries}

\theoremstyle{definition}

\crefname{Exm}{example}{examples}

\crefname{Rem}{remark}{remarks}

\newtheorem{Alg}{Algorithm}[section]
\crefname{Alg}{algorithm}{algorithms}

\newtheorem{Def}{Definition}[section]
\crefname{Def}{defintion}{defintions}

\usepackage{tikz}
\usepackage{graphicx}
\usetikzlibrary{matrix,arrows,shapes,trees,fit}
\usetikzlibrary{calc,positioning,scopes,backgrounds}
\newif\ifoutertikznest
\outertikznestfalse
\usetikzlibrary{decorations.markings}

\tikzstyle{lw} = [line width=3pt]

\usepackage{pgfplots}
\pgfplotsset{compat=newest}

\definecolor{box_white}{cmyk}{0.0361,0.0251,0.0166,0}
\definecolor{text_black}{cmyk}{0.7979,0.7417,0.6916,0.6554}
\definecolor{tui_orange_dark}{cmyk}{0,0.6,1,0} 
\definecolor{tui_orange_light}{cmyk}{0.0000,0.0876, 0.1474, 0.0157}
\definecolor{tui_green_dark}{cmyk}{1,0,0.5,0.2}
\definecolor{tui_green_light}{cmyk}{0.0576,0.0041, 0.0000, 0.0471} 
\definecolor{tui_blue_dark}{cmyk}{1.0000,0.5000,0.0000,0.6000}
\definecolor{tui_blue_light}{cmyk}{0.0920,0.0440,0.0000,0.0196}
\definecolor{tui_red_dark}{cmyk}{0.0000,1.0000,1.0000,0.2000}
\definecolor{tui_red_light}{cmyk}{0.0000,0.1107,0.1107,0.0078}

\pgfplotscreateplotcyclelist{tui_dark}{
{thick,tui_orange_dark,mark=*,mark options={thin,scale=0.75, fill=tui_orange_dark!50!white}},
{thick,tui_green_dark,mark=square*,mark options={thin,scale=0.75, fill=tui_green_dark!50!white}},
{thick,tui_blue_dark,mark=triangle*,mark options={thin,scale=0.75, fill=tui_blue_dark!50!white}},
{thick,tui_red_dark,mark=diamond*,mark options={thin,scale=0.75, fill=tui_red_dark!50!white}},
}

\pgfplotsset{
    cycle list name=tui_dark,
}

\newcommand{\linePlot}[7]{
\begin{figure}[ht]
\begin{center}
\begin{tikzpicture}[inner frame sep=0,tight background]
\begin{axis}[
legend style={
font=\small,
at={(0.5,-0.20)},
anchor=north,
legend columns=-1,
legend cell align=left,
draw=none,
fill=none
},
xlabel style={yshift=0.5em},
ylabel style={yshift=-0.4em},
width=0.95\columnwidth,
height=0.70\columnwidth,
enlargelimits=0.0,
xlabel={\textbf{#1}},
ylabel={\textbf{#2}},
grid style={dashed,tui_blue_light},
grid=both,
font=\small,
xticklabel style={text height=1.5ex},
ytick style={draw=none},
xtick style={draw=none},
#3
]
#4
\legend{#5};
\end{axis}
\end{tikzpicture}
\caption{#6}\label{#7}
\end{center}
\end{figure}
}

\newcommand{\linePlotTwo}[9]{
\begin{figure}[ht]
\begin{center}
\begin{tikzpicture}[inner frame sep=0,tight background]
\begin{axis}[
axis y line*=left,
legend style={
font=\small,
at={(0.5,-0.20)},
anchor=north,
legend columns=-1,
legend cell align=left,
draw=none,
fill=none
},
xlabel style={yshift=0.5em},
ylabel style={yshift=-1.2em},
width=0.95\columnwidth,
height=0.70\columnwidth,
enlargelimits=0.0,
xlabel={\textbf{#1}},
ylabel={\textbf{#2}},
font=\small,
xticklabel style={text height=1.5ex},
#3
]
#4
\end{axis}
\begin{axis}[
axis y line*=right,
legend style={
font=\small,
at={(0.5,-0.20)},
anchor=north,
legend columns=-1,
legend cell align=left,
draw=none,
fill=none
},
xlabel style={yshift=0.5em},
ylabel style={yshift=0.2em},
width=0.95\columnwidth,
height=0.70\columnwidth,
enlargelimits=0.0,
xticklabels={,,}
ylabel={\textbf{#5}},
font=\small,
xticklabel style={text height=1.5ex},
#6
]
#7
\end{axis}
\end{tikzpicture}
\caption{#8}\label{#9}
\end{center}
\end{figure}
}

\usepackage{datatool}
\usepackage{numprint}
\DTLsetseparator{,}

\definecolor{tui_orange_dark}{cmyk}{0,0.6,1,0} 
\definecolor{tui_orange_light}{cmyk}{0.0000,0.0876, 0.1474, 0.0157}
\definecolor{tui_green_dark}{cmyk}{1,0,0.5,0.2}
\definecolor{tui_green_light}{cmyk}{0.0576,0.0041, 0.0000, 0.0471} 
\definecolor{tui_blue_dark}{cmyk}{1.0000,0.5000,0.0000,0.6000}
\definecolor{tui_blue_light}{cmyk}{0.0920,0.0440,0.0000,0.0196}
\definecolor{tui_red_dark}{cmyk}{0.0000,1.0000,1.0000,0.2000}
\definecolor{tui_red_light}{cmyk}{0.0000,0.1107,0.1107,0.0078}
\title{Sparsity Order Estimation from a Single Compressed Observation Vector}
\author{Sebastian Semper$^*$,~\IEEEmembership{Student Member,~IEEE},
		    Florian R\"{o}mer,~\IEEEmembership{Senior Member,~IEEE},
		    Thomas Hotz, 
		    Giovanni DelGaldo,~\IEEEmembership{Member,~IEEE}
\thanks{
This work was partially supported by the Deutsche Forschungsgemeinschaft (DFG) project CoSMoS (grant GA 2062/2-1) and the 
Carl-Zeiss Foundation under the postdoctoral scholarship project ``EMBiCoS''.
}
\thanks{
The authors 
S.~Semper and T.~Hotz are with
Technische Universität Ilmenau,
Institute of Mathematics,
Ilmenau, Germany. 
The authors F.~Roemer and G.~Del~Galdo are with
Technische Universität Ilmenau,
Institute for Information Technology,
Ilmenau, Germany. 
The author G.~Del~Galdo is also with Fraunhofer Institute for Integrated Circuits (IIS),
Ilmenau, Germany.
}
\thanks{$*$ corresponding author}
}

\usepackage[]{hyperref} 
\newtheorem{xremark}{Remark}

\begin{document}

\maketitle

%%%%%%%
%%% Abstract
%%%%%%%
\noindent
{\bf {\bf \slshape Abstract \symbol{124}}
We investigate the problem of estimating the unknown degree of sparsity from compressive measurements without the need to carry out a sparse recovery step. While the sparsity order can be directly inferred from the effective rank of the observation matrix in the multiple snapshot case, this appears to be impossible in the more challenging single snapshot case. We show that specially designed measurement matrices allow to rearrange the measurement vector into a matrix such that its effective rank coincides with the effective sparsity order. In fact, we prove that matrices which are composed of a Khatri-Rao product of smaller matrices generate measurements that allow to infer the sparsity order. Moreover, if some samples are used more than once, one of the matrices needs to be Vandermonde. These structural constraints reduce the degrees of freedom in choosing the measurement matrix which may incur in a degradation in the achievable coherence. We thus also address suitable choices of the measurement matrices. In particular, we analyze Khatri-Rao and Vandermonde matrices in terms of their coherence and provide a new design for Vandermonde matrices that achieves a low coherence.
}

%%%%%%%
%%% EDICS (keywords): needed for the final version
%%%%%%%
% \begin{center}
%   {\bf EDICS:}
%    
% \end{center}

\IEEEpeerreviewmaketitle

%%% Subsequent chapters, I usually put them into separate files
\section{Introduction}
The Compressive Sensing (CS) paradigm states that 
a signal $\bm{x} \in \C^N$ can be acquired using $M < N$ linear measurements, 
provided that it possesses a sparse representation with 
$K \ll N$ coefficients in a suitable basis $\bm{A} \in \C^{N \times N}$ \cite{candes2006uncertprinc}.
This observation has sparked considerable interest in many engineering disciplines
devoted to signal acquisition, including imaging (e.g., MRI \cite{lustig2007sparsemri}, X-Ray
\cite{sidky2008image}), Radar \cite{ender2010compressiveradar,baraniuk2007compressive} and 
spectrum sensing \cite{mishali2009blind,tropp2010beyond}, among many others.

The number of non-zeros $K$ is also referred to as the sparsity order or sparsity degree.
Even though the number of measurements required $M$ quite clearly depends on the sparsity order, the latter
is typically not known when acquiring a signal. Therefore, one usually determines $M$ according to some upper bound
$K \leq K_{\rm max}$ which must be known in order to design the measurement. However, since $K$ may vary quite
significantly between observation periods, such an approach can lead to systems that are too conservative, i.e.,
they take more measurements than necessary. 
Therefore, being able to estimate and monitor the sparsity order of a signal 
would be an attractive feature of a measurement device. It would allow to adjust the number of measurements to
the current sparsity of the signal which might vary in time (e.g., with the number of transmissions in cognitive
radio \cite{yucek2009survey} or with the complexity of an image in CS-based image acquisition \cite{duarte2008single}).
Moreover, knowing the sparsity order allows to improve
reconstruction algorithms by tuning algorithm-specific parameters such as the regularization parameter $\lambda$
in LASSO-type techniques (which is related to $K$ \cite{thrampoulidis2015asymptotically}) or stopping criteria in greedy techniques such as the
Orthogonal Matching Pursuit \cite{tropp2007signal}.

\subsection{Related Work}
%\TODO{text still needs some fine-tuning}
%
%% haven't used so far:
% Lan Zhang, Zhang, Guo, ``Sparsity estimation im image compressive sensing'' (ISCAS2012)
% -> very specific for images. using image complexity measures to relate this to sparsity.
% 
% Austin et. al. ``On the relation between sparse recosntruction and parameter estimation with model order selection'' (JSTSP 2010)
% -> "A close relation between sparse signal reconstruction and parameter estimation with model order selection has been discussed in [4], where the sparsity-promoting regularization parameter (which influences the model order of the sparse solution) is chosen according to classical information criteria"

Due to the prominent role the sparsity order plays in sparse signal recovery, the lack of knowledge
of the sparsity order has been recognized as a fundamental gap between theory and practice
\cite{ward2009cscrossval,malioutov2008seqobserv}. Still, existing literature on this subject is quite scarce
and very recent. Early papers on this subject have proposed to employ 
sequential measurements \cite{malioutov2008seqobserv} and cross-validation
type techniques \cite{boufounos2007sparse,ward2009cscrossval} where sequential reconstructions of the signal are considered.
Similarly, \cite{wang2012soecograd} shows that the sparsity order can be estimated from
the reconstruction, stating bounds on the number of measurements that are required  for this
step. However, the bounds are only found numerically, and the reconstruction process
involves cumbersome optimization problems. As we show, this can be avoided
by estimating the sparsity order directly based on the compressed observations.

A different approach is taken in \cite{lopes2012soe}, where the authors show that a specifically
tailored measurement procedure which consists of a Cauchy and a Gaussian distributed
measurement matrix allows to estimate a continuous measure of sparsity given by the ratio
of squared one- and two-norm of the signal. However, this measure is not equal to the sparsity order. In
fact, it is continuous and hence needs to be rounded to an integer number (which is not discussed
in \cite{lopes2012soe}). Moreover, the measurement process is very restrictive since
the distribution of the measurement matrices is pre-specified.

The authors in \cite{ravazzi2016gammasoe} propose to use sparse sensing matrices since these
allow to infer the degree of sparsity of a signal from the degree of sparsity of the measurement.
The resulting estimator has a very low complexity. However, it is only approaching the
true sparsity in the large system limit and hence not applicable to lower-dimensional problems.
Moreover, the proposed measurement matrices incur a certain performance degradation at the reconstruction stage due to the somewhat higher coherence, which is imposed by the restriction to sparse matrices.

A link between sparsity order estimation and rank estimation was put forward in  \cite{sharma2014soewbcr} for the multiple measurement vector (MMV) problem, which we have also studied in our own prior work,
both for the stationary case \cite{lavrenko2014ettsoe,lavrenko2015subnyq} as well as the case of time-varying support for block-stationary signals \cite{lavrenko2015timevarsupp}.
However, these approaches require a certain stationarity in the support patterns as well as a temporal variation in the coefficients of the sparse representation to create
linearly independent observations. This limits their applicability in many practical problems.

\subsection{Contributions}

This paper introduces a method for estimating the sparsity order of a signal from a single
snapshot of the compressive measurement. In particular, we 
propose to consider rearrangements of the observation vector into a matrix and show under which
conditions the rank of this matrix coincides with the sparsity order of the unknown signal.
Thereby, the sparsity order can be estimated by applying any known rank estimation scheme \cite{stoica2004model}.
Since there exist many efficient algorithms which estimate the model order in presence of perturbations such
as noise, this approach allows us to handle noisy measurements as well as the case of approximate
sparsity as well.

As we show, the proposed approach only requires the measurement matrix to possess a Khatri-Rao structure,
which leaves considerable room for optimizing their choice.
Moreover, in the case of overlapping blocks, one of the factors needs to be a Vandermonde matrix.
Therefore, the second part of the manuscript discusses the design of measurement matrices in presence
of these structural constraints. 
To facilitate the sparsity order estimation as well as the sparse reconstruction, the factors need to
possess a low coherence. 
As we show, in the presence of a Khatri-Rao structure, it is best to optimize
the factors of the Khatri-Rao product independently. 
We therefore investigate the coherence of Vandermonde matrices and propose
a new design algorithm, which efficiently constructs Vandermonde matrices with low coherence. 
Additionally, we  derive simple upper and lower bounds for the resulting coherence of this algorithm.
%\TODO{details?}

Compared to the earlier conference version \cite{roemer2014singsnapsoe}, this paper 
contains more rigorous versions of the main Theorems on the sparsity order estimation,
along with their proofs (which had to be partially omitted in the conference version due to lack of space).
Moreover, the optimization of the measurement matrix with respect to the structural constraints
was added.

%\TODO{details}
%In this work, we elaborate on the previous conference contribution \cite{single_snap_soe} where a column-wise Kronecker product as the measurement matrix is used to extract the sparsity order directly from a single compressed measurement snapshot. Here the measurement is split up into $k$ subvectors of length $\ell$. Those are stacked into a matrix as its columns. Then the rank of this resulting matrix coincides with the desired support size of the underlying signal. 

\subsection{Notation}
Capital bold face letters denote matrices, whereas lowercase boldface letters stand for vectors. We use $[N]$ for $N \in \N$ as a shorthand for the set $\{1,\dots,N\}$. The set $\Supp(\bm x)$ is the set of indices where the vector $\bm x$ is non-zero, if $\Supp(\bm x) = s$ we call $\bm x$ $s$-sparse. The vector $\bm x_S$ subscribed with an index set $S$ is the restriction of $\bm x$ to the indices $S$ whereas the matrix $\bm X_S$ denotes the matrix $X$ restricted to the columns indexed by $S$. Moreover, $\Norm{\cdot}_p$ stands for the $\ell_p$ norm where $p$ is omitted if it is clear from the context. 
% Discuss special case p->0
The operators $\bm{A} \otimes \bm{B}$ and $\bm{A} \diamond \bm{B}$ refer to the Kronecker product and the Khatri-Rao (column-wise Kronecker) product, respectively.
Finally, $\Rk^*\bm{A}$ represents the Kruskal rank of a matrix $\bm{A}$, i.e.\ the largest natural number $r$ so that every set of $r$ columns of $\bm{A}$ is linearly independent.
\section{Sparsity Order Estimation}
%%%

\subsection{Data Model Setup}

 We consider the following discrete sparse recovery problem
\begin{align}
    \bm{b} = \bm{A}\cdot \bm{S}\cdot \bm{x} + \bm{n},\label{orig_model}
\end{align}
 where $\bm{x} \in \C^N$ is $K$-sparse, $\bm{S} \in \C^{N\times N}$ represents the sparsifying basis,
and $\bm{A} \in \C^{m\times N}$ with $m\ll N$ is the measurement matrix.
Since it is assumed that $\bm{S}$ is a basis (and hence invertible) we can limit ourselves to the case
$\bm{S} = \bm{I}_N$ without loss of generality if $\bm{A}$ can be chosen freely since we can always
replace $\bm{A}$ by $\bm{\bar{A}} = \bm{A} \cdot \bm{S}^{-1}$ to account for an $\bm{S} \neq \bm{I}_N$.

Moreover, for clarity we consider the noise-free case $\bm{n}=\bm{0}$ first. The role of additive noise
is discussed in Section~\ref{sec:SSR_noisy}.

In order to estimate the sparsity order $K$ directly from $\bm{b}$, we propose to consider rearrangements
of $\bm{b}$ into a matrix. Specifically, we divide $\bm{b}$ into $k$ blocks $\bm{b}_i$ of length $\ell$
where the $i$-th block is given by
\begin{align}
   \bm{b}_i = \begin{bmatrix} b_{1+p\cdot(i-1)}, & b_{2+p\cdot(i-1)}, & \ldots, b_{\ell+p\cdot(i-1)} \end{bmatrix}^\trans \in \C^{\ell} \label{eqn:def_blocks}
\end{align}
for $i=1,2, \ldots, k$, where $p$ specifies by how many samples consecutive blocks are advancing. 
It is clear from~\eqref{eqn:def_blocks} that for $p=\ell$ the blocks $\bm{b}_i$ do not overlap whereas for $p<\ell$ the overlap grows with decreasing $p$, up
to the case of maximum overlap for $p=1$. Moreover, since $m$ samples are available, the parameters $\ell, p, k, m$ should satisfy $\ell+p\cdot(k-1) = m$.

The blocks $\bm{b}_i$ can be used to form the columns of a matrix $\bm{B} = \begin{bmatrix} \bm{b}_1, & \bm{b}_2, & \ldots, & \bm{b}_\ell \end{bmatrix} \in \C^{\ell \times k}$. Moreover, we define submatrices $\bm A_i$ by selecting the corresponding rows from $\bm A$ such that
\begin{align}
   \bm{b}_i = \bm A_i \cdot \bm x. \label{eqn:def_Ai}
\end{align}
The main idea of the proposed approach is to show that for a suitably chosen $\bm{A}$, we have $\Rk \bm{B} = K$ for any $K$-sparse $\bm{s}$ 
and therefore, the sparsity order
can be inferred from the rank of $\bm{B}$. Since the required structure differs for the case of overlapping and non-overlapping blocks,
we treat these two cases in separate subsections.

\subsection{Sparsity order estimation for non-overlapping blocks}

The following theorem summarizes the conditions on $\bm{A}$ to facilitate the sparsity order estimation for the case of non-overlapping blocks ($p = \ell$).

\begin{Thm}[\cite{roemer2014singsnapsoe}]\label{without_overlap}
	For $k$ blocks of the measurement $\bm b$ of length $\ell = m/k$ and any $r \leq \Min(k,\ell)$ the following statements are equivalent.
	\enum{
		\item\label{wthout_ovrl_1} For all $s \leq r$ and all $\bm{x}$ with $\Abs{\Supp\{\bm{x}\}} = s$ we have $\Rk \bm{B} = s$.
		\item\label{wthout_ovrl_2} $ \bm A = \bm \Phi \diamond \bm \Psi$ for some $\bm \Phi \in \C^{k \times N}$, $\bm \Psi \in \C^{\ell \times N}$ with $\Rk^* \bm \Phi \geq r$ and $\Rk^* \bm \Psi \geq r$.
	}
\end{Thm}
\begin{proof}
Assume \ref{wthout_ovrl_1} holds and consider an arbitrary but fixed $1$-sparse $\bm x$ with $\Supp(\bm x) = \{q\}$ for some $q \in [N]$. Now \ref{wthout_ovrl_1} implies that $\Rk \bm B = 1$. On the other hand, denoting the $q$-th column of $\bm A_i$ by $\bm a_{i,q}$, we have
	\[\bm b_i = \bm A_i \cdot \bm x = \bm a_{i,q} \cdot x_q \Text{for} i \in [k].\]
Because all $\bm b_i$ must be linearly dependent, we get 
	\[\bm a_{i,q} = \varphi_{i,q} \cdot \bm\psi_q \Text{for} i \in [k],\]
for some non-zero $\bm \psi_q \in \C^\ell$ and some $\varphi_{i,q} \in \C$. In terms of the columns of $\bm A$ this means $\bm a_q = \Vectorize(\bm \psi_q \bm \varphi_q^\trans) = \bm{\varphi}_q \otimes \bm{\psi}_q$ for $\bm \varphi_q = (\varphi_{1,q} , \dots , \varphi_{k,q})^\trans \in \C^k$. Since $q$ was arbitrary we conclude that $\bm A = \bm \Phi \diamond \bm \Psi$, where $\bm \Phi = (\bm \varphi_i)_{i \in [N]} \in \C^{k \times N}$ and  $\bm \Psi = (\bm \psi_i)_{i \in [N]} \in \C^{\ell \times N}$. 

Now consider any $s$-sparse $\bm x$ for $s \leq r$. Then from the structure of $\bm A$ and $\bm B$, we get that
	\Eq{\bm B = \bm \Psi \Diag(\bm x) \bm \Phi^\trans,}{struct_B}
since $p = \ell$. This can also be restricted to $S = \Supp(\bm x)$ and then reads as
	\[\bm B = \bm \Psi_S \Diag(\bm x_S) \bm \Phi_S^\trans.\]
Now seeking a contradiction, assume that $\Rk^* \bm \Phi < r$. This means that there is a set $T \subset [N]$ of size $r$ such that $\bm \Phi_T$ is rank deficient. Now set $\bm x = \sum_{i \in T} \bm e_i$, which yields $\bm B = \bm \Psi_T \bm \Phi_T^\trans$. The fact that $\bm B$ has rank strictly smaller than $r$ implies the desired contradiction. An analogue argument holds if we assume $\Rk^* \bm \Psi < r$.

For the opposite direction, assume \ref{wthout_ovrl_2} and again consider equation (\ref{eq:struct_B}) for some $s$-sparse $\bm x$ with support $S$. Knowing that $\bm \Phi$ and $\bm \Psi$ have Kruskal rank at least $s$, we also have $\Rk \left[\bm \Psi_S \Diag(\bm x_S) \right] = s$ and $\Rk \bm \Phi_S = s$. This implies that $\Rk \bm B = s$, because the two matrices are a rank factorization of $\bm B$.
\end{proof}

In other words, Theorem~\ref{without_overlap} states that in the case of non-overlapping blocks, we can obtain the sparsity order $K$ 
from the rank of $\bm{B} \in \C^{\ell \times k}$ if and only if the measurement matrix $\bm{A}$ possesses a Khatri-Rao structure.
Moreover, the theorem also shows that the highest sparsity order we can estimate is given by 
$K_{\rm max} = {\rm min}(\Rk^*(\bm{\Phi}), \Rk^*(\bm{\Psi})) \leq \min(k,\ell)$. Note that  in the case of non-overlapping blocks
we have $k \cdot \ell = m$ which means $K_{\rm max} \leq \min(k,\frac{m}{k}) \leq \sqrt{m}$ 
and we should choose $k$ and $\ell$ close to $\sqrt{m}$ to maximize $K_{\rm max}$.

%Rephrasing the above statement one observes that the measurement matrix must obey a column-wise Kronecker structure where its factors' Kruskal Ranks are the delimiting quantities in what sparsity order $s$ can maximally be estimated since $\Min(k,\ell) \leqslant \sqrt{m}$. This also implies that $k$ and $\ell$ should be chosen close to $\sqrt{m}$ to maximize $\Min(k,\ell)$.

\subsection{Sparsity order estimation for overlapping blocks}

As we have seen for non-overlapping blocks, if we let $\bm{A}$ be a Khatri-Rao product of two matrices of equal size, the maximum sparsity order that can be estimated is $K_{\rm max} \leqslant \sqrt{m}$. This bound is tight, iff $m$ is a square number $m = q^2$; then we can set $k = \ell = q$. If we let blocks overlap, the size of the matrix $\bm{B}$ grows, which allows to estimate larger sparsity orders $K_{\rm max}$. However, depending on the overlap, additional constraints on $\bm{A}$ have to be posed. This is formalized by the following theorem.

\begin{Thm}[\cite{roemer2014singsnapsoe}]\label{with_overlap}
	For $k$ overlapping blocks of length $\ell$, block advance $p$ and any $r \leqslant \Min(k,\ell)$ the following statements are equivalent:
	\enum{
		\item\label{wth_ovrl_1} For all $s \leqslant r$ and all $s$-sparse $\bm x$ we have $\Rk \bm B = s$.
		\item\label{wth_ovrl_2} $\bm A$ consists of the first $m$ rows of $\bm V \diamond \bm \Psi$ with $\bm \Psi \in \C^{p \times N}$ being arbitrary, $\bm V \in \C^{\lceil m/p \rceil \times N}$ being a Vandermonde matrix such that the matrix $\hat{\bm \Psi}$ originating from restricting $\bm V_{\lceil \ell/p \rceil} \diamond \bm \Psi$ to its first $\ell$ rows has Kruskal rank $r$ and the matrix $\bm V$ restricted to its first $k$ rows also has Kruskal rank $r$.
		
		%with $\Rk^* \bm \Psi \geqslant r$ and $\Rk^* \bm V \geqslant r$.
	}
\end{Thm}
\begin{proof}
Given \ref{wth_ovrl_1}, we again consider some $\bm x \in \C^N$ with $\Supp(\bm x) = \{ q \}$ such that
	\[\bm B = x_q [\bm a_{1,q},\dots,\bm a_{k,q}].\]
Since $\bm B$ has rank $1$, we know that its columns, which are also the columns of the blocks $\bm A_i$, fulfill
\begin{align}
	\bm a_{i,q} = z_{i,q} \bm a_{i-1,q} \Text{for} i \in \{2,\dots,k\},\label{scaling_1} \\
	\bm a_{i+1,q} = z_{i+1,q} \bm a_{i,q} \Text{for} i \in \{1, \dots, k-1\}.\label{scaling_2} \\
\intertext{for appropriate $z_{i,q}$ and still denoting the $q$-th column of $\bm A_i$ by $\bm a_{i,q}$. Moreover, because of the overlap we also get}
	(\bm a_{i,q})_j = (\bm a_{i-1,q})_{p+j} \Text{for} j \in [\ell-p] + 1 \label{overlap_1}\\
	(\bm a_{i+1,q})_j = (\bm a_{i,q})_{p+j} \Text{for} j \in [\ell-p] \label{overlap_2}\\
\intertext{and $i \in [k-1]$. Combining \eqref{scaling_1},\eqref{scaling_2} and \eqref{overlap_2} from above yields}
	z_{i+1,q} \left( \bm a_{i,q} \right)_{j} = z_{i,q} \left( \bm a_{i-1,q} \right)_{p+j} \Text{for} j \in [\ell - p] \\
\intertext{and $i \in [k-1]$. From \eqref{overlap_1} we follow, that his has to be true for all $j \in [\ell - p]$, so}
	z_{i,q} = z_{i+1,q} = z_q \Text{for} i \in [k-1]
\end{align}
for some $z_q$ independent of $i$. Summarizing, we get
	\[\bm a_{i,q} = z_q  \bm a_{i-1,q} = z_q^2 \bm a_{i-2,q} = \dots = z_q^{i-1} \bm a_{1,q}.\]
This means that we can only choose $z_q$ and because of the overlap the first $p$ elements of $\bm a_{1,q}$ freely. Now setting $\left(\bm a_{1,q}\right)_{[p]} = \bm \psi_q$ with $\bm \psi_q \in \C^{p}$ then implies the Khatri-Rao structure $\hat{\bm A} = \bm V \diamond \bm \Psi$, where we get $\bm A$ by restricting $\hat{\bm A}$ to its first $m$ rows and $\bm \Psi = [\bm \psi_1, \dots, \bm \psi_N]$. The Vandermonde structure of the first factor reads as
	\[\bm V = \begin{bmatrix} 1 & \dots & 1 \\ z_{1} & \dots & z_{N} \\ \vdots & \ddots & \vdots \\ z^{\lceil m/p \rceil - 1}_{1} & \dots & z^{\lceil m/p \rceil - 1}_{N} \end{bmatrix},\]
if we repeat the reasoning above for any $q \in [N]$. 

For the Kruskal ranks of the involved matrices we deduce from the Khatri-Rao structure and the overlap that
	\[\bm B = \hat{\bm \Psi} \cdot \Diag(\bm x) \cdot \bm V_k^\trans,\]
where $\hat{\bm \Psi}$ consists of the first $\ell$ rows of $\bm V_{\lceil \ell/p \rceil} \diamond \bm \Psi$ and $\bm V_{\lceil \ell/p \rceil}$ and $\bm V_k$ are the restrictions of $\bm V$ to its first $\lceil \ell/p \rceil$ and $k$ rows respectively. Now consider the above equation for any $s$-sparse $\bm x$. As in the last proof, we argue that the factors involved in the Khatri-Rao product must fulfill the Kruskal rank conditions imposed on them in the statement of \ref{wth_ovrl_2} for $\bm B$ to have rank $s$.

Now given \ref{wth_ovrl_2}, we proceed similarly to the proof of Theorem \ref{without_overlap}, but this time, we get for an $s$-sparse $\bm x$ with $s \leqslant r$ that
	\[\bm B = \hat{\bm \Psi} \cdot \Diag(\bm x) \cdot \bm V_k^\trans,\]
where $\bm V_{\lceil \ell/p \rceil}$ and $\bm V_k$ are the restrictions of $\bm V$ to its first $\lceil \ell/p \rceil$ and $k$ rows respectively. Because of the conditions on $\hat{\bm \Psi}$ and $\bm V$ we deduce as before that $\bm B$ has rank $s$. 
\end{proof}

The assumptions in \ref{wth_ovrl_2} in Theorem \ref{with_overlap} on $\bm V$ and $\bm \Psi$ seem rather strict and technical. In the following we will show, that they are fulfilled given some simple criteria. If we consider the Vandermonde matrix $\bm V$ it is enough to require that the $z_{1}, \dots, z_N$ are pairwise distinct for $\bm V$ to have maximal Kruskal rank. For a square Vandermonde matrix $\bm V \in \C^{n \times n}$ with generating elements $z_1 , \dots , z_n$ reads as
	\[\det(\bm V) = \Prod{1 \leqslant i < j \leqslant n}{}{(z_i - z_j)}.\]
This means that a Vandermonde matrix with more rows than columns and pairwise distinct generating elements always has full column rank. Moreover, if the matrix has more columns than rows, and pairwise distinct generating elements, every square submatrix has full rank and as such the whole matrix has full Kruskal rank.

Now turning to the arbitrary factor $\bm \Psi \in \C^{p \times N}$, we can choose it in general position and in the following we show that this is sufficient for the requirements of Theorem \ref{with_overlap}. To this end, fix some $r$-sparse $\bm x \in \C^N$ with support set $R$, a Vandermonde matrix $\bm V \in \C^{m \times N}$ which has full Kruskal rank and consider the Khatri-Rao product $\bm V \diamond \bm \Psi \in \C^{m \cdot p \times N}$. Now assume that
\[\bm m = \left(\bm V \diamond \bm \Psi \right) \cdot \bm x = \bm 0.\]
If we reshape the vector $\bm m$ according to the block construction in Theorem \ref{with_overlap} to
\[\bm M = \bm \Psi \Diag(\bm x) \bm V^\trans = \widetilde{\bm \Psi} \bm V^\trans,\]
which is a rank decomposition of $\bm M$ and this implies that for non-zero $\bm x$ the matrix $\bm M$ is non-zero. 

Note that the special case $p = 1$ in Theorem~\ref{with_overlap} implies that the entire sensing matrix $\bm A$ is a Vandermonde matrix with rescaled columns. In the context of harmonic retrieval, the mapping from $\bm b$ to $\bm B$ is also known as spatial smoothing~\cite{shan1985spatialsmooth} and is applied as a preprocessing step for subspace-based estimators in order to decorrelate coherent signals.

\begin{xremark}\label{parameter_choice}
Regarding the choice of the parameters we observe that $K_{\rm max} \leq \min(k,\ell)$ where $k$ and $\ell$ satisfy $(k-1)\cdot p + \ell = m$. Therefore, for a given block advance $p$, the number of blocks $k$ is equal to $k=\frac{m-\ell}{p}+1$. This means that to maximize $K_{\rm max}$ we should choose $\ell$ as the closest integer to $\frac{m+p}{p+1}$, which leads to $k\approx \ell$. Obviously, larger amounts of overlap (corresponding to smaller values of $p$) lead to a higher maximum sparsity order $K_{\rm max}$ where the maximum overlap case $p=1$ corresponds to $K_{\rm max} = \lfloor\frac{m+1}{2}\rfloor$. This shows that there is a fundamental tradeoff between the sparsity order estimation (SOE) stage and the sparse recovery stage: while a larger amount of overlaps improves the SOE capability, it leads to a more rigidly structured measurement matrix with a higher coherence, which is detrimental to the sparse recovery step. The achievable coherence is analyzed in more depth in Section~\ref{sec:analysis}.
\end{xremark}

\subsection{Sparsity Order Estimation in the presence of noise} \label{sec:SSR_noisy}

In the presence of additive, say Gaussian, noise as in equation (\ref{orig_model}) the entries of the matrix $\bm B$ are disturbed with a Gaussian noise matrix $\bm N$ as well, i.e.
\begin{equation}\label{eqn:b_noisy}
\hat{\bm B} = \bm B + \bm N.
\end{equation}
It is easy to see that the rank of $\hat{\bm{B}}$ is full with probability $1$. It is still possible to determine the ``effective'' rank for a model like \eqref{eqn:b_noisy} if the statistics of the additive noise are known.
Interestingly, although $\bm N$ is a reshaped version of the vector $\bm n$ in (\ref{orig_model}) with possibly the same noise sample at multiple positions (depending on the overlap), we can show that if the noise samples in $\bm{n}$ are i.i.d., the noise matrix $\bm{N}$ 
is ``white'' in the sense that $\mathbb{E}\left(\bm{N} \bm{N}^\herm\right) = C\cdot \bm{I}_\ell$.
To this end, let us assume the elements in $\bm{n}$ have zero mean and variance one
and let us define the selection matrices 
%
%\[\bm S_i = \left[\bm 0_{\ell \times r_i} , \bm I_{\ell \times \ell} , \bm 0_{\ell \times t_i} \right]\]
\[\bm{J}_i = \begin{bmatrix} \bm 0_{\ell \times p(i-1)} , & \bm I_{\ell \times \ell} , & \bm 0_{\ell \times (N - \ell - p(i-1))} \end{bmatrix} \in \R^{\ell \times m},\]
%
%with $r_i = p(i-1)$ and $t_i = N - l - p(i-1)$ for $i \in [k]$. 
which satisfies $\bm{J}_i \cdot \bm{J}_i^\herm = \bm{I}_{\ell\times \ell}$. 
Then one can rewrite the measurement and noise vectors as
\[\bm b_i = \bm{J}_i \cdot \bm{b} \Text{and} \bm n_i = \bm{J}_i \cdot \bm n,\]
where $i=[k]$. Now we calculate
\begin{align}
 \mathbb{E} \left( \bm N \cdot \bm N^\herm \right) 
  & = \mathbb{E} \left( \sum_{i = 1}^k\bm{J}_i \cdot \bm n \cdot (\bm{J}_i \cdot \bm n)^\herm\right) \notag \\
	& = \sum_{i = 1}^k\bm{J}_i \cdot \mathbb{E} \left(\bm n \cdot \bm{n}^\herm \right) \cdot \bm{J}_i^\herm \notag \\
	& = \sum_{i=1}^k\bm{J}_i \cdot \bm{J}_i^\herm  =k \cdot \bm{I}_\ell.
\end{align}
	%
%%
%which in the case of $i = j$ gives
%%
%\[\mathbb{E} \left( \Sum{\kappa = 1}{k}{\Abs{(\bm S_\kappa \bm n)_i}^2}\right) = \sigma^2\]
%%
%for some appropriate $\sigma^2 > 0$. For $1 \leqslant i \neq j \leqslant \ell$, we the summation adds products of the $i$-th and $j$-th row of $\bm N$. But because of the block advance $p > 1$, the any row in the matrix $\bm N$ has the same entry of $\bm n$ more than once. So all factors in $(\bm S_\kappa \bm n)_i \cdot (\bm S_\kappa \bm n)^\herm_{~~ j}$ are independent for any $\kappa$, $i$ and $j$, which yields
%%
%\[\mathbb{E} \left( \Sum{\kappa = 1}{k}{(\bm S_\kappa \bm n)_i \cdot (\bm S_\kappa \bm n)^\herm_{~~ j}}\right)_{i,j}^{\ell,\ell} = 0.\]
%%
%So in case of additive white noise in the measurements, this translates to additive white noise in the entries of $\hat{\bm B}$.

As we have shown, we need to determine the rank of the matrix $\bm{B}$
in presence of noise according to \eqref{eqn:b_noisy} (where in our case, the additive
noise is white). This task is known as model order selection and a number
of efficient algorithms are available.
%Hence, we can apply model order selection procedures without prewhitening to estimate the underlying sparsity order. There are several choices available, 
Examples include information-theoretic criteria such as MDL, AIC, BIC (see \cite{stoica2004model} for a survey), the \emph{Eigenvalue Threshold Test} (ETT)\cite{lavrenko2014ettsoe} or the \emph{Exponential Fitting Test} (EFT) presented in \cite{quinlan2007eft}. We make use of the latter for our numerical experiments in section~\ref{sec:numerics}, because it is derived specifically for models
%the requirement for this algorithm is that the model is 
disturbed by additive white Gaussian measurement noise.
\section{Sensing matrix design}\label{sec:analysis}
As we have seen in the last section, sparsity order estimation can be achieved via
rank estimation of a matrix obtained by rearranging the measurement vector, provided
that the sensing matrix obeys certain structural constraints. Firstly, it has to be a Khatri-Rao
product of two smaller matrices and secondly, in the case of overlapping blocks,
one of the blocks has to be a Vandermonde matrix.

In this section, we analyze the implications this 
particular sensing matrix structure has
for the design of the measurement matrix, with
a particular focus on the Vandermonde matrices.

\subsection{Khatri-Rao structured measurement matrix optimization}\label{SOE_SSR}

As we have seen in Theorem~\ref{without_overlap} in the case of no overlap one is able to recover the sparsity order by using a Khatri-Rao structured sensing matrix, whose factors have to fulfill the condition of having a high Kruskal rank. 
This condition is difficult to optimize for, since the Kruskal rank is hard to compute.
However, the following inequality links it to the coherence
\begin{equation}\label{coherence}
\mu(\bm A) = \Max\limits_{1 \leqslant i < j \leqslant m} \frac{\ScPr{\bm a_i}{\bm a_j}}{\Norm{\bm a_i}\Norm{\bm a_j}},
\end{equation}
namely for any matrix $\bm{X} \in \C^{p \times q}$, $p<q$ we have
\[\Rk{}^* \bm X \geq \frac{1}{\mu(\bm X)}.\]
Therefore, for the Kruskal rank of the factors to be high, we aim for matrices with low coherence. 
These are desirable also from the viewpoint of the subsequent sparse recovery step,
which in general works better the lower the coherence of the sensing matrix is. In fact,
it has been shown that the under-determined system of equations $\bm{b} = \bm{A} \cdot \bm{s}$ 
has a unique solution for $K$-sparse vectors $\bm{s}$ if $K < K_{\rm max}$ with
\begin{align}
K_{\rm max} = \frac{1}{2}\left(1+\frac{1}{\mu(\bm{A})}\right). \label{eqn:kmax_ssr}
\end{align}
Note that Basis Pursuit (BP) as well as Orthogonal Matching Pursuit (OMP) are able to achieve this bound in the sense that they can recover any $K$-sparse $\bm{s}$ in the noise-free case as long as $K<K_{\rm max}$.
For this reason we minimize the coherence of the measurement matrix $\bm{A}$. 

Since $\bm{A}$ is a Khatri-Rao product, it might seem advantageous to take this structure into account when optimizing the coherence of the sensing matrix using the two factors.
In the real valued setting, one can use packing arguments in projective matrix spaces of rank $1$
to optimize the coherence of Khatri-Rao products.
As a result one obtains that the best coherence of a Khatri-Rao product $\bm{A} = \bm{\Phi} \diamond \bm{\Psi} \in \R^{m \times N}$
is achieved
if $\bm{\Phi} \in \R^{m_1 \times N}$ and $\bm{\Psi} \in \R^{m_2 \times N}$ 
contain repeated columns according to $\bm{\Phi} = \hat{\bm{\Phi}} \otimes \bm{1}_{1\times N_2}$
and $\bm{\Psi} = \bm{1}_{1 \times N_1} \otimes \hat{\bm{\Psi}}$
where $\bm{\hat{\Phi}} \in \real^{m_1 \times N_1}$, $\bm{\hat{\Psi}} \in \real^{m_2 \times N_2}$
with $m = m_1 \cdot m_2$ and $N = N_1 \cdot N_2$.
In this case we get
$\bm{A} = \bm{\Phi} \diamond \bm{\Psi} = \hat{\bm \Phi} \otimes \hat{\bm \Psi}$ so that $\bm{A}$
is actually Kronecker structured.
Due to the repeating columns, this implies that both $\bm{\Phi}$ and $\bm{\Psi}$ have Kruskal rank one,
i.e. Khatri-Rao products of minimal coherence have factors that are not suitable
for sparsity order estimation. Instead, one should optimize the coherence of the 
factors $\bm{\Phi}$ and $\bm{\Psi}$ independently.

This also has a positive effect on the coherence of $\bm{A}$ since we have the trivial upper bound
%This means that the only theoretical insight we get is the trivial coherence bound for Khatri-Rao products, which reads as
%
\[\mu(\bm A) = \mu(\bm \Phi \diamond \bm \Psi) \leq \mu(\bm \Phi \otimes \bm \Psi) = \mu(\bm \Phi) \mu(\bm \Psi).\]
%
%So in order to improve upon the coherence of the product, the best we can do is carry out the optimization of the factors separately. On the other hand this coincides with the first requirement on the sensing matrix' factors concerning the Kruskal rank.
Therefore, the remainder of this section discusses the minimization of coherence of the matrices $\bm{\Phi}$ and $\bm{\Psi}$.

In the case of no overlap these matrices can be arbitrary as there are no further structural requirements from the sparsity
order estimation method. 
We can therefore apply any method for coherence minimization from the literature.
In fact, there are already many iterative construction methods like in~\cite{tropp2008grasspack} and in coding theory this problem arises in constructing good spherical codes. An overview of theory and algorithms can be found in~\cite{zoernlein2015sphercode}.

This means that the only regime where we can improve upon existing results is in the case of overlap during the construction of the matrix $\bm B$ as in Theorem~\ref{with_overlap}. Here, one factor obeys a Vandermonde structure. As long as the generating elements of the matrix $\bm V$ are pairwise different it is well known that this matrix has maximal Kruskal rank. But since we also need a low coherence for efficient recovery to happen, the following section contains a more thorough analysis of this problem.

\subsection{Coherence minimization for Vandermonde matrices}
%%As we have seen above we need matrices that obey structural requirements, which are suitable matrices during the measurement process for sparsity order estimation.

%We noticed above that Vandermonde matrices play a key role during SOE. 
Vandermonde matrices are very rigidly structured, i.e. they can be generated from the knowledge of their first row.
\begin{Def}
	Let $\mathcal{V}^{n \times m}(\C)$ be the space of $n\times m$ Vandermonde matrices then it holds that the non-linear mapping
	\Eq{\bm \nu_n \D \C^m \rightarrow \mathcal{V}^{n \times m} \Text{with} \bm z 
		\mapsto \begin{pmatrix}
				z_1 & \dots & z_m \\
				z_1^2 & \dots & z_m^2 \\
				\vdots & \vdots & \vdots \\
				z_1^n & \dots & z_m^n 
				\end{pmatrix}
	}{vander_bij}
	is bijective. So whenever we think about a Vandermonde matrix $\bm V$ it suffices to consider its first row.
\end{Def}
In the following, we will call the $z_i$ the generating elements of $\bm V$. 

This strong requirement on the structure of elements in $\mathcal{V}^{n \times m}$ will allow us to derive an explicit term for the inner product of two columns of a Vandermonde matrix that does not involve any summation and just depends on the amplitudes and phases of the generating elements of the respective columns. This is key, because of the definition of the coherence $\mu( \bm A)$ of a matrix given in \eqref{coherence}.

Further on we will be using the lemma below a couple of times, which calculates the value of the discontinued geometric series.
\begin{Lem}
	 For any $q \in \C \setminus \{1\}$ it holds that
		\Eq{\Sum{k = 1}{n}{q^k} = \frac{q - q^{n+1}}{1 - q} = S_n(q)}{ab_sum}
	In the case $q = 1$ the above formula yields
		\Eq{\lim\limits_{q \rightarrow 1} S_n(q) = \lim\limits_{q \rightarrow 1} \frac{q - q^{n+1}}{1 - q} = \lim\limits_{q \rightarrow 1} \frac{1 - (n+1)q^{n}}{-1} = n.}{sum_lhosp}
\end{Lem}
\begin{proof}
	This follows from a straightforward application of L'H{\^o}pital's rule.
\end{proof}
The following theorem contains all information needed to derive an algorithm for Vandermonde matrices with low coherence and bounds thereof. We treat the general case for two arbitrary Vandermonde columns but also the special case where the first row of the matrix is a subset of the complex unit circle. Furthermore it defines upper and lower envelope functions (see Figure~\ref{vanderHullc1c2}) for the inner product of two Vandermonde columns depending only on the phase shift between two generating elements for fixed amplitudes.
\begin{Thm}\label{lambda_facts}
	If we consider two vectors \[\bm z_1 = (c_1 e^{i\phi_1},\dots,c_1^n e^{i n \phi_1})^\trans\] and \[\bm z_2 = (c_2 e^{i\phi_2},\dots,c_2^n e^{i n \phi_2})^\trans,\] where $c_1$, $c_2$ are positive real numbers and $\phi_1 \leqslant \phi_2$ with $\phi_1, \phi_2 \in [0,2\pi]$ and define the function 
		\Eq{\lambda(c_1,c_2,\phi) = \frac{\Abs{\ScPr{\bm z_1}{\bm z_2}}^2}{\Norm{\bm z_1}^2 \Norm{\bm z_2}^2}}{lambda_def}
	with $\phi = \phi_2 - \phi_1$, then the following statements hold.
	\enumrom{
		\item\label{fcts_lmb_gen}
			\[\lambda(c_1,c_2,\phi) = \dfrac{\dfrac{1 + c_1^{2n}c_2^{2n} - 2 c_1^n c_2^n \cos[n\phi]}{(1 - c_1^{2n})(1 - c_2^{2n})}}{\dfrac{1 + c_1^2c_2^2 - 2 c_1 c_2 \cos\phi}{(1-c_1^2)(1-c_2^2)}}.\]
		\item\label{fcts_lmb_11}
			\[\lambda(1,1,\phi) = \frac{\sin^2\left(\frac 12 n \phi\right)}{n^2 \sin^2\left(\frac 12 \phi\right)}.\]
		\item\label{fcts_lmb_sym} We have symmetry along the unit circle in $\C$. In other words
			\[\lambda(c_1,c_2,\phi) = \lambda(c_1^{-1},c_2^{-1},\phi) = \lambda(c_1,c_2,-\phi).\]
		\item\label{fcts_lmb_hulls} For $c_1 \neq 1$ and $c_2 \neq 1$ we can bound $\lambda$ by
			\[\kappa(c_1,c_2,\phi) \leqslant \lambda(c_1,c_2,\phi) \leqslant \eta(c_1,c_2,\phi),\]
		where $\kappa$ and $\eta$ are defined as
		\[\kappa(c_1,c_2,\phi) = \dfrac{\dfrac{(1 - c_1^n c_2^n)^2}{(1 - c_1^{2n})(1 - c_2^{2n})}}{\dfrac{1 + c_1^2c_2^2 - 2 c_1 c_2 \cos\phi}{(1-c_1^2)(1-c_2^2)}}, \]
		and
		\[\eta(c_1,c_2,\phi) = \begin{cases}
										\dfrac{1}{n^2\sin^2\left(\frac 12 \phi\right)} \Text{for} c_1 = c_2 = 1 \\
										\dfrac{\dfrac{(1 - c_1^n c_2^n)^2}{(1 - c_1^{2n})(1 - c_2^{2n})}}{\dfrac{1 + c_1^2c_2^2 - 2 c_1 c_2 \cos\phi}{(1-c_1^2)(1-c_2^2)}} \Text{otherwise.}
			                      \end{cases}
		\]
		Moreover, the bounds satisfy for $c_1,c_2 \in (0,1]$
			\begin{align*}
				\lambda(c_1,c_2,\phi) &= \kappa(c_1,c_2,\phi) \\ 
				&\Text{if and only if} \phi = k \cdot (2\pi)/n \\
				&\Text{for some} k \in \{0,\dots,n\}, \\
				\lambda(c_1,c_2,\phi) &= \eta(c_1,c_2,\phi) \\
				&\Text{if and only if} \phi = \pi/n + k \cdot (2\pi)/n \\
				&\Text{for some} k \in \{1,\dots,n-1\}, \\
				\kappa(c_1,c_2,\phi) &= \kappa(1/c_1,1/c_2,\phi) \\
				& \Text{for} \phi \in [0,2\pi], \\
				\eta(c_1,c_2,\phi) &= \eta(1/c_1,1/c_2,\phi) \\
				& \Text{for} \phi \in (0,2\pi).
			\end{align*}
		\item\label{fcts_lmb_limit} For any $c > 0$ and $c_1 > 0$ it holds that
			\[\lim\limits_{c_1 \rightarrow \infty} \lambda(c/c_1,c_1,\phi) = 0 \Text{for} \phi \in [0,2\pi].\]
		\item\label{fcts_lmb_roots} Roots of $\lambda(c_1,c_2,\cdot)$ exist if and only if $c_1c_2 = 1$ and these roots are $\phi = k \cdot (2\pi/n)$ for $k \in \{1,\dots,n-1\}$, i.e. there are exactly $n-1$ roots.
	}	
\end{Thm}

\linePlot{$\phi$}{$ $}{
ymax=1.1
}{
\addplot+[mark=,densely dotted] table [col sep=comma,y=x1]{matlab/data/vanderHull2_5.csv};
\addplot+[mark=,densely dotted] table [col sep=comma,y=x3,x=phi]{matlab/data/vanderHull2_5.csv};
\addplot+[ultra thick,mark=] table [col sep=comma,y=x2]{matlab/data/vanderHull2_5.csv};
}{
lower envelope $\eta$, upper envelope $\kappa$,$\lambda$
}
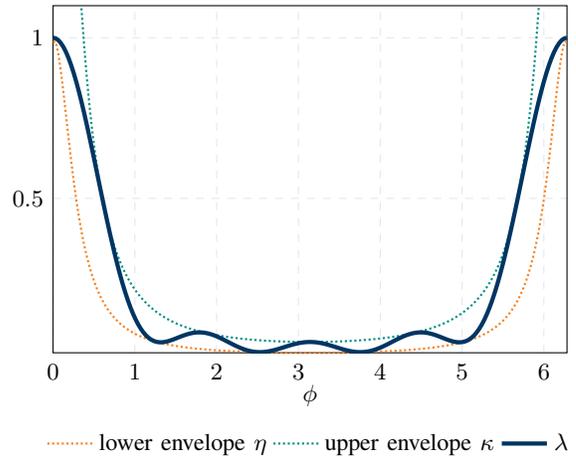
{The envelope curves for $\lambda(c_1,c_2,\cdot)$ for the case $c_1 = c_2^{-1} = 1.15$ and $n=5$.}{vanderHullc1c2}

\begin{proof}~\\
	\enumrom{
		\item We calculate the inner product of $\bm z_1$ and $\bm z_2$ with the law of cosines and the formula in \eqref{eq:ab_sum}, which yields
			\begin{align*}
				\Abs{\ScPr{\bm z_1}{\bm z_2}}^2 &= \Abs{\Sum{k=1}{n}{c_1^k e^{ik\phi_1}\cdot c_2^k e^{-ik\phi_2}}}^2\\
				&= \Abs{\Sum{k=1}{n}{(c_1c_2)^ke^{ik\phi}}}^2 \\
				&= \Abs{\dfrac{c_1 c_2 e^{i \phi} - (c_1 c_2)^{n+1}e^{i \phi(n+1)} }{1 - c_1 c_2 e ^{i\phi}}}^2 \\
				&= c_1^2c_2^2 \frac{1 + c_1^{2n}c_2^{2n} - 2 c_1^n c_2^n \cos(n\phi)}{1 + c_1^2c_2^2 - 2 c_1 c_2 \cos\phi}
			\end{align*}
		In this case we would have to consider L'H{\^o}pital's rule for $c_1 c_2 = 1$ and $\phi = 0$, which would yield that $c_1 c_2 e^{i \phi} = 1$ and thus $\Abs{\ScPr{\bm z_1}{\bm z_2}}^2 = n^2$. Because we want to derive an expression for the coherence of a matrix, we still need to divide by the norm of the vectors $\bm z_1$ and $\bm z_2$. Now using \eqref{eq:ab_sum} again we get
			\begin{align*}
					\lambda(c_1,c_2,\phi) &= \frac{\Abs{\ScPr{\bm z_1}{\bm z_2}}^2}{\Norm{\bm z_1}^2 \Norm{\bm z_2}^2} \\
					&= \dfrac{\dfrac{1 + c_1^{2n}c_2^{2n} - 2 c_1^n c_2^n \cos[n\phi]}{(1 - c_1^{2n})(1 - c_2^{2n})}}{\dfrac{1 + c_1^2c_2^2 - 2 c_1 c_2 \cos\phi}{(1-c_1^2)(1-c_2^2)}},
			\end{align*}
		where we have to also consider L'H{\^o}pital's rule for $c_1 = 1$ or $c_2 = 1$. Then we would get $\Norm{\bm z_1}^2 = n$ or $\Norm{\bm z_2}^2 = n$ respectively.
		\item If we make use of \ref{fcts_lmb_gen} and take the appropriate limits, this yields
		\Eq{
			\lambda(1,1,\phi) = \frac{1 - \cos(n\phi)}{n^2(1 - \cos\phi)} = \frac{\sin^2\left(\frac 12 n \phi\right)}{n^2 \sin^2\left(\frac 12 \phi\right)}  ,
		}{lambda_11}
		\item We simply manipulate $\lambda(c_1^{-1},c_2^{-1},\phi)$ by factoring out the proper powers of $c_1$ and $c_2$ and get $\lambda(c_1,c_2,\phi)$.
% 			\begin{align}
% 				\lambda(c_1^{-1},c_2^{-1},\phi) &= \frac{1 + c_1^{-2n}c_2^{-2n} - 2 c_1^{-n} c_2^{-n} \cos[n\phi]}{1 + c_1^{-2}c_2^{-2} - 2 c_1^{-1} c_2^{-1} \cos\phi} \frac{(1-c_1^{-2})(1-c_2^{-2})}{(1 - c_1^{-2n})(1 - c_2^{-2n})} \notag \\
% 					&= \frac{c_1^{-2n}c_2^{-2n}\left(1 + c_1^{2n}c_2^{2n} - 2 c_1^n c_2^n \cos[n\phi]\right)}{c_1^{-2}c_2^{-2}\left(1 + c_1^2c_2^2 - 2 c_1 c_2 \cos\phi\right)} \frac{c_1^{-2}c_2^{-2}(1-c_1^2)(1-c_2^2)}{c_1^{-2n}c_2^{-2n}(1 - c_1^{2n})(1 - c_2^{2n})}\notag  \\
% 					&= \lambda(c_1,c_2,\phi).\label{eq:lambda_sym_calc}
% 			\end{align}
		%
		%
		%
		\item The statement follows from the fact that $-1 \leqslant \cos x  \leqslant 1$ for all $x \in \R$ and the locations where $\Abs{\cos x} = 1$.
		\item If we assume that $c_1 \cdot c_2 = c < 1$ then we can suppose without loss of generality that $c_1 \leqslant c_2$. If we hold $c$ fixed, we can rewrite $\lambda$ as
			\[\lambda(c_1,c_2,\phi) = \lambda(c/c_2,c_2,\phi) =\]
			\[= \frac{\left(\cos(n\phi) - c^n\right)^2 + \sin^2(n\phi)}{\left(\cos\phi - c\right)^2 + \sin^2\phi} \frac{(1-(c/c_2)^2)(1-c_2^2)}{(1 - (c/c_2)^{2n})(1 - c_2^{2n})} = \]
			\[= \frac{\left(\cos(n\phi) - c^n\right)^2 + \sin^2(n\phi)}{\left(\cos\phi - c\right)^2 + \sin^2\phi} \frac{1 + c^2\left(1 - \frac{1}{c_2^2}\right) - c_2^2}{1 + c^{2n}\left(1 - \frac{1}{c_2^{2n}}\right) - c_2^{2n}}.\]
		If we now take the limit for growing $c_2$ and fixed $\phi$, we get:
			\begin{align}
				\lim\limits_{c_2 \rightarrow \infty} &\lambda(c/c_2,c_2,\phi) = 0 \\
				\Text{for all} &\phi \in [0,2\pi] \Text{and} c<1. \notag
			\end{align}
		\item Let us first consider the case $c_1 \cdot c_2 < 1$. First we rewrite $\lambda$ in another form with the trigonometric identity $\sin^2x + \cos^2x = 1$ by which we obtain
			\Eq{
				\lambda(c_1,c_2,\phi) = \dfrac{\dfrac{\left(\cos(n\phi) - (c_1c_2)^n\right)^2 + \sin^2(n\phi)}{(1-c_1^2)(1-c_2^2)}}{\dfrac{\left(\cos\phi - c_1c_2\right)^2 + \sin^2\phi}{(1 - c_1^{2n})(1 - c_2^{2n})}}.
			}{lambda_sin}
		
		From \eqref{eq:lambda_sin}, we know that $\lambda$ is zero if and only if
			\[\sin(n\phi) = \cos(n\phi) - (c_1c_2)^n = 0.\]
		Let us assume that $\sin(n\phi) = 0$, which implies that $\phi = k \cdot (\pi/n)$. Now we know that $\Abs{\cos(k\pi)} = 1$ and the only case that $\cos(n\phi) - (c_1c_2)^n = 0$ is when $c_1 c_2 = 1$ what we explicitly excluded. In other words for $c_1c_2$ different than $1$, we can never find two vectors $z_1$ and $z_2$ which are orthogonal to each other. 
	}
\end{proof}
Some results from above are visualized in Figure~\ref{vanderHullc1c2}. The last Theorem allows us to derive a construction for orthogonal Vandermonde matrices, which can be considered a generalization of the fact that the Fourier matrix is orthogonal. The Fourier matrix is, depending on the definition, a scaled Vandermonde matrix, where the generating elements are placed on a regular grid on the unit circle, i.e. $\exp(2 \pi i k/n)$, $k \in [n-1]$. The corollary below additionally allows a uniform phase shift of the generating elements.
\begin{Cor}\label{vand_coh}
	Let $z_1,\dots,z_n$ be chosen with absolute value $1$ such that
		\[\arg\left(z_i z_j^*\right) = \frac{(i-j)2\pi}{n} \Text{for} 1 \leqslant j < i \leqslant n.\]
	Then for the matrix $\bm V = \nu_n(z_1,\dots,z_n)$ it follows that
		\[\mu(\bm V) = 0.\]
\end{Cor}
\begin{proof}
	This is a direct consequence of \ref{fcts_lmb_roots} in Theorem \ref{lambda_facts}.
\end{proof}
But in the scenario typical for compressed sensing the involved matrices have more columns than rows and thus it is impossible for them to be orthogonal. Hence we are left with the objective to minimize the coherence as far as possible. The following algorithm, which is also one of the central novelties in this work, makes use of Theorem~\ref{lambda_facts} to achieve this.

The general geometric idea is twofold. First, the upper envelope $\kappa$ suggests that a large angular distance between generating elements yields a lower inner product, since it is monotonically decreasing on $(0, \pi)$. To maximize the mutual angular difference between each pair of generating elements, we need to place them on a regular angular grid on $[0, 2\pi]$. Second, facts \ref{fcts_lmb_sym} and \ref{fcts_lmb_limit} suggest that the inner product of two columns is minimized if their generating elements have amplitudes that are reziprocals of each other or, geometrically speaking, are reflections on the unit circle in $\C$ of each other. 

Combining the above two observations yields the algorithm as stated below. Loosely speaking it places the angularly nearest neighbors among the generating elements on different sides of the complex unit circle and the ratio between the amplitudes is the parameter $c$. 
\begin{Alg}[Vandermonde matrices]\label{alg:vander_better}$~$\\
	\textsc{Input:}
		\items{
			\item Dimensions $n$ and $m$
			\item Positive constant $c$
		}
	\textsc{Output:}
		\items{
			\item A Vandermonde matrix $\bm V \in \C^{n \times m}$
		}
	\textsc{Procedure:}
		\enum{
			\item Set $\hat{m} = 2\lceil m/2 \rceil$, $c_1 = c$ and $c_2 = 1/cy$.
			\item Set $\phi_k = 4\pi\frac{k-1}{\hat{m}}$ and $z_k = c_1 \cdot e^{i \cdot \phi_k}$ for $k \in [\hat{m}/2]$.
			\item Set $\phi_k = 4\pi\frac{k-1}{\hat{m}} + 2\pi\frac{1}{\hat{m}}$ and $z_k = c_2 \cdot e^{i \cdot \phi_k}$ for $k \in \{\hat{m}/2 + 1,\dots,\hat{m}\}$.
			\item Return $\bm V = \nu(z_1,\dots,z_m)$.
		}
\end{Alg}
The algorithm above still depends on the input parameter $c$, but the optimal value can be computed with a simple bisection algorithm and depends only on $n$ and $m$. This allows it to be stored in a persistent lookup table.

Additionally we would like to derive bounds on the coherence. This allows efficient estimation of the algorithm's performance before its application. The bounds exploit the fact that the above algorithms output is deterministic and thus the upper and lower envelope $\kappa$ and $\eta$ can be applied. The following two Theorems deal with the upper and lower bounds respectively.

A lower bound on the coherence can be used together with equation \eqref{eqn:kmax_ssr} to derive lower bounds on the number of measurements required in order to achieve a recovery guarantee for a given sparsity order $K_{\rm max}$.
\begin{Thm}\label{lwr_vander_bnd}
	Given $n$, $m > n$ and $c \in (0,1)$ and the output $\bm V$ of \ref{alg:vander_better} for inputs $n$, $m$ and $c$ it holds that
		\begin{align}
			\mu(\bm V) \geqslant \Max\big\{ &\sqrt{\kappa(c,c,4\pi/m)}, \notag \\
			&\sqrt{\kappa(c,1/c,2\pi/m)}\big\}.
		\end{align}
\end{Thm}
\begin{proof}
	From Theorem~\ref{lambda_facts} we know that $\lambda(c_1,c_2,\phi) \geqslant \kappa(c_1,c_2,\phi)$ for all $c_1$, $c_2$ and $\phi$. Algorithm~\ref{alg:vander_better} produces a matrix $\bm V$ so that the minimum angle between two elements in its first row is $2\pi/m$. Moreover those elements are located either on the circle with radius $c$ or the circle with radius $1/c$. If two elements share the same circle, their minimum angle between each other is $4\pi/m$. If not they enclose the angle $2\pi/m$. This yields the estimate
		\begin{align}
			\mu(\bm V)^2 \geqslant \Max\big\{& \kappa(1/c,1/c,4\pi/m),\kappa(c,c,4\pi/m), \notag \\
			&\kappa(c,1/c,2\pi/m),\kappa(1/c,c,2\pi/m)\big\},
		\end{align}
	which becomes what we asserted when we consider the symmetry of $\kappa$ as stated in fact \ref{fcts_lmb_hulls} in Theorem~\ref{lambda_facts} and that $\lambda$ is the square of the inner product of two vectors with Vandermonde structure.
\end{proof}
Conversely we can use equation \eqref{eqn:kmax_ssr} and upper coherence bounds to derive upper bounds on the number of measurements required to guarantee successful recovery in scenarios with a certain level of sparsity.
\begin{Thm}\label{upp_vnd_bnd}
	If we define the function $u \D (0,1] \times \N \rightarrow \R$ by
	\[u(c,m) = \begin{cases}
						\max\{\eta(c,c,4\pi/m),\eta(1/c,c,2\pi/m)\} \\\Text{for} m < 2n, \\
						\max\{\eta(c,c,4\pi/m),\lambda(1/c,c,2\pi/m)\} \\ \Text{for} 2n \leqslant m \leqslant 4n, \\
						\max\{\lambda(c,c,4\pi/m),\lambda(1/c,c,2\pi/m)\} \\ \Text{for} m > 4n,
	              \end{cases}
	\]
	and call Algorithm~\ref{alg:vander_better} with parameters $n$, $m$ and $c$, then for its output $\bm V$ it follows that:
		\[\mu(\bm V) \leqslant \sqrt{u(c,m)}.\]
\end{Thm}
\begin{proof}
	Suppose we have an output $\bm V$ of Algorithm~\ref{alg:vander_better} for given $c$, $n$ and $m$. We remind ourselves that the first row of $\bm V$ is placed on two circles with radii $c$ and $1/c$. Then we know that the angle between two elements in the first row of $\bm V$ which are not on the same circle is an integer multiple of $2\pi/m$. Elements on the same circle enclose an angle which is an integer multiple of $4\pi/m$.

	\noindent\textbf{Case 1:} Suppose that $m < 2n$. From the monotony of $\eta$ and its symmetry
		\[\eta(c_1,c_2, \pi - \phi) = \eta(c_1,c_2, \pi + \phi) \Text{for} \phi \in (0,\pi)\]
	we obtain
	\[\eta(c,c,4\pi/m) \geqslant \eta(c,c,k \cdot 4\pi/m) \geqslant \lambda(c,c,k \cdot 4\pi/m)\]
	for $k \in \{1,\dots,\lfloor m/2 \rfloor\}$ and 
	\[\eta(1/c,c,2\pi/m) \geqslant \eta(1/c,c,k  2\pi/m)\]
	for $k \in \{1,\dots,m\}$. Since the above estimates address all possible values $\lambda$ can take for $\bm V$ and the maximum over all values $\lambda$ attains the square of the coherence of $\bm V$, we have concluded the proof for $m < 2n$.
	
	\noindent\textbf{Case 2:} Suppose now that $4n \geqslant m \geqslant 2n$. From the monotony of $\eta$ we can again follow that
		\begin{align}
			\lambda(1/c,c,2\pi/m) & > \eta(1/c,c,\pi/n) \geqslant \eta(1/c,c,k \cdot \pi/n) \notag\\
			&\geqslant \lambda(1/c,c,k \cdot 2\pi/m) \notag \\
			&\Text{for} k \in \{2,\dots,m\}.
		\end{align}
	Moreover, from $4n > m \geqslant 2n$ we get
		\[\frac{\pi}{n} < \frac{4\pi}{m} \leqslant \frac{2\pi}{n}\]
	and estimate 
		\begin{align}
			\eta(c,c,4\pi/m) & \geqslant \eta(c,c,k \cdot 4\pi/m) \geqslant \lambda(c,c,k \cdot 4\pi/m) \notag \\
			&\Text{for} k \in \{1,\dots,\lfloor m/2 \rfloor\}.
		\end{align}
	These estimates address all possible values for $\lambda$ and so the statement follows in this case.
	
	\noindent\textbf{Case 3:} Suppose now that $m > 4n$. From the monotony of $\lambda$ on $[0,2\pi/n]$ and the fact that $4\pi/m$ and $2\pi/m$ are left of the first touching point of $\lambda$ and $\eta$, namely $\pi/n$, we can estimate
		\begin{align}
			\lambda(c,c,4\pi/m) & \geqslant \lambda(c,c,k_1 \cdot 4\pi/m) \geqslant \lambda(c,c,\pi/n) \notag \\
			& = \eta(c,c,\pi/n) \geqslant \eta(c,c,k_2 \cdot 4\pi/m) \notag \\
			&\geqslant \lambda(c,c,k_2 \cdot 4\pi/m),
		\end{align}
	for $k_1 \in \{1,\dots,\lfloor m/(4n) \rfloor\}$ and $k_2 \in \{\lfloor m/(4n) \rfloor + 1, \dots, \lfloor m/2 \rfloor\}$. With a similar argument we estimate
		\begin{align*}
			\lambda(1/c,c,2\pi/m) &\geqslant \lambda(1/c,c, k_3 \cdot 2\pi/m) \geqslant \lambda(1/c,c,\pi/n) \\
			&= \eta(1/c,c,\pi/n) \geqslant \eta(1/c,c,k_4 \cdot 2\pi/m) \\
			&\geqslant \lambda(1/c,c,k_4 \cdot 4\pi/m),
		\end{align*}
	for $k_1 \in \{1,\dots,\lfloor m/(2n) \rfloor\}$ and $k_2 \in \{\lfloor m/(2n) \rfloor + 1, \dots, m\}$. This again estimates all values $\lambda$ ever takes and we conclude the proof.	
\end{proof}

\section{Simulations}\label{sec:numerics}
% include the parameter files
\begin{filecontents*}{matlab/data/simu_vand_params.csv}
name,value 
num_t,2000
num_n,512
num_s,6
\end{filecontents*}
\DTLloaddb{vand_params}{matlab/data/simu_vand_params.csv}
\begin{filecontents*}{matlab/data/simu_soe_params.csv}
name,value 
num_err_prob,0.00500000
num_s,8
num_s_max,40
num_t,1000
num_n,512
\end{filecontents*}
\DTLloaddb{soe_params}{matlab/data/simu_soe_params.csv}
This chapter is dedicated to empirical investigations for showcasing the performance of the proposed methods. To this end, we consider measurement scenarios as in \eqref{orig_model}. In order to generate the ground truths $\bm x$, we draw the amplitudes on the support of $\bm x \in \C^N$ for $N = 512$ i.i.d. from the set $\{\pm 1 \pm i, \pm 1 \mp i\}$ according to a uniform distribution. We carry out all our simulations with $m$ noisy measurements where the components of the additive noise vector $\bm n \in \C^m$ were drawn independently from a zero mean circularly symmetric Gaussian distribution with variance $\sigma^2$, where the total SNR is then defined as $\sigma^{-2}$. Depending on the situation at hand, the construction of $\bm A$ changes accordingly, which is chosen once for each scenario and kept fixed when sampling the ground truths for a certain scenario type and size. We carry out $2000$ trials for each scenario and level of SNR.
\subsection{Sparsity Order Estimation}
To depict the performance of the proposed method for SOE, we simulate the procedure for various combinations of parameters $m$ and $p$. Here, we set the sparsity order of the ground truth $\bm x \in \C^N$ to $8$. After setting the number of measurements $m$ and the overlap $p$, we select $k$ and $\ell$ in order to maximize the size of the reshaped matrix $\bm B$ for optimal performance according to Remark \ref{parameter_choice}. The Vandermonde factors in the columnwise Kronecker product are constructed by Algorithm~\ref{alg:vander_better} and the unstructured factors are drawn once from a Gaussian ensemble, which means that all elements are drawn independently from a Gaussian distribution with variance $1$. Both factors in each of the occuring instances of a columnwise Kronecker product have normalized columns, resulting in the product having normalized columns as well.

To study the influence of the overlap we vary the parameter $p$ for fixed signal size $N$ and number of measurements $m$ while adapting $k$ and $\ell$ appropriately. For the estimation of the effective rank of $\bm B$, we apply the Empirical Eigenvalue-Threshold Test (ETT)~\cite{lavrenko2014ettsoe} to our scenario with a target false rejection rate of $0.005$. The training stage of this model order selection method is the computationally most time consuming part, but this procedure only has to be done once in advance.

\linePlot{SNR}{mean of estimated sparsity order}{
legend style={
legend columns=2,
},
xmin=0,
xmax=50
}{
\addplot table [col sep=comma,x=SNR,y=soe_mean]{matlab/data/soe_121_1.csv};
\addplot table [col sep=comma,x=SNR,y=soe_mean]{matlab/data/soe_121_2.csv};
\addplot table [col sep=comma,x=SNR,y=soe_mean]{matlab/data/soe_121_4.csv};
\addplot table [col sep=comma,x=SNR,y=soe_mean]{matlab/data/soe_121_0.csv};
}{
$p = 1$ max. overlap,$p = 2$,$p = 4$,$p = \ell$ no overlap}
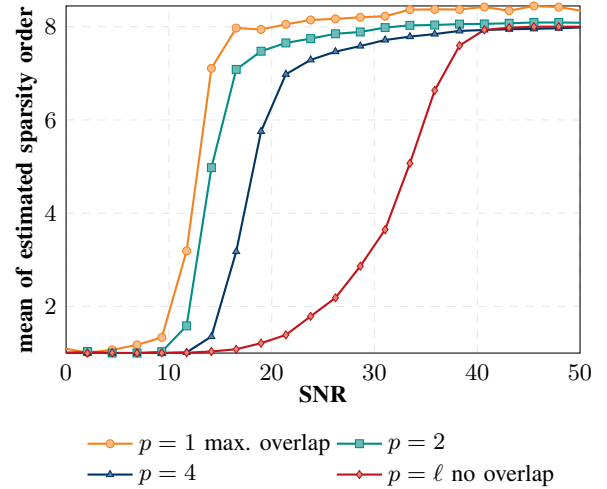
{Influence of overlap $p$ on the sparsity order estimation stage with $m = 121$ measurements.}{comp_overlap_soe}

Figure~\ref{comp_overlap_soe} displays the empirical mean of the estimated sparsity order across all trials for a varying levels of the SNR. As expected for the pure task of SOE the case $p = 1$ displays the best performance, because $\Min\{k,\ell\}$ is maximal. For increasing $p \geqslant 2$, including the case $p = \ell$, the phase transitions happens at an increasing level of SNR and the transition itself is not as sharp as in the cases of overlap.

To study a more realistic processing pipeline where the true sparsity order is not known in advance but just a more or less sharp upper bound, we combine the process of SOE with the reconstruction using OMP by feeding the estimated sparsity order into the algorithm as a runtime parameter. This is compared to an ``unguided'' OMP reconstruction, where we assume to only know upper bounds $K_{\mathrm{Max}} = 40$ and $K_{\mathrm{Max}} = 20$. Then OMP is run for $K$ steps.

\linePlotTwo{SNR}{$l_2$-norm error}{
xmin=20,
xmax=45,
ymode=log,
}{
\addplot table [col sep=comma,y=omp_gauss_MSE_max1,x=SNR]{matlab/data/soe_169_0_X.csv};\label{plot_one1}
\addplot table [col sep=comma,y=omp_gauss_MSE_max2,x=SNR]{matlab/data/soe_169_0_X.csv};\label{plot_one2}
\addplot table [col sep=comma,y=omp_soe_MSE,x=SNR]{matlab/data/soe_169_0.csv};\label{plot_two}
}
{
succ rate
}{
legend style={
legend columns=1,
},
xmin=20,
xmax=45,
}
{
\addlegendimage{/pgfplots/refstyle=plot_one1}\addlegendentry{$l_2$ error without SOE and $K_{\mathrm{Max}} = 40$ (left axis)}
\addlegendimage{/pgfplots/refstyle=plot_one2}\addlegendentry{$l_2$ error without SOE and $K_{\mathrm{Max}} = 20$ (left axis)}
\addlegendimage{/pgfplots/refstyle=plot_two}\addlegendentry{$l_2$ error with SOE (left axis)}
\addplot+[tui_red_dark] table [col sep=comma,y=soe_succ,x=SNR]{matlab/data/soe_169_0.csv};
\addlegendentry{success rate of SOE (right axis)}
}
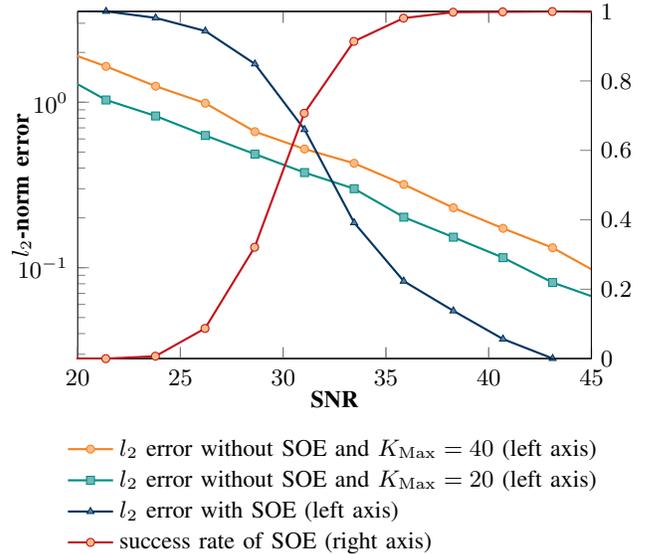
{Influence of SOE on reconstruction quality when using estimated sparsity order as parameter for OMP.}{soe_raw_comp}

As one can see in Figure~\ref{soe_raw_comp}, where the success rate of SOE and the resulting $l_2$ error after reconstruction are displayed in a combined plot, as soon as SOE starts to work, the reconstruction error decreases significantly below the one of the ``unguided'' reconstruction. So in this case the additional computational effort for SOE results in a significant increase in reconstruction precision.

As a last study, we investigate the reconstruction performance of the matrices that obey the structural constraints imposed by our means of SOE. To this end, we simulate a scenario with a priori known sparsity order to remove the influence of the SOE procedure on the reconstruction process. Here we expect differing performance since for decreasing $p$, the highly structured Vandermonde block in the sensing matrix becomes more dominant thus increasing the overall coherence of $\bm A$. The results are depicted in Figure \ref{comp_overlap_omp}. Here one can see that the case of no overlap $p = \ell$ comes close to the Gaussian measurement matrices despite the imposed columnwise Kronecker structure. Moreover one notices that the Vandermonde factor in the columnwise Kronecker product for the case of $p < \ell$ is the reason for a worse reconstruction performance.

\linePlot{SNR}{$l_2$-norm error}{
legend style={
legend columns=2,
},
xmin = 0,
xmax = 20,
ymode=log,
}{
\addplot table [col sep=comma,x=SNR,y=omp_nosoe_MSE]{matlab/data/soe_64_1.csv};
\addplot table [col sep=comma,x=SNR,y=omp_nosoe_MSE]{matlab/data/soe_64_2.csv};
\addplot table [col sep=comma,x=SNR,y=omp_nosoe_MSE]{matlab/data/soe_64_4.csv};
\addplot table [col sep=comma,x=SNR,y=omp_nosoe_MSE]{matlab/data/soe_64_0.csv};
\addplot[very thick,black,densely dotted] table [col sep=comma,x=SNR,y=omp_gauss_MSE_true]{matlab/data/soe_64_4.csv};
}{
$p = 1$ max. overlap,$p = 2$,$p = 4$,$p = \ell$ no overlap, unstructured Gaussian matrix}
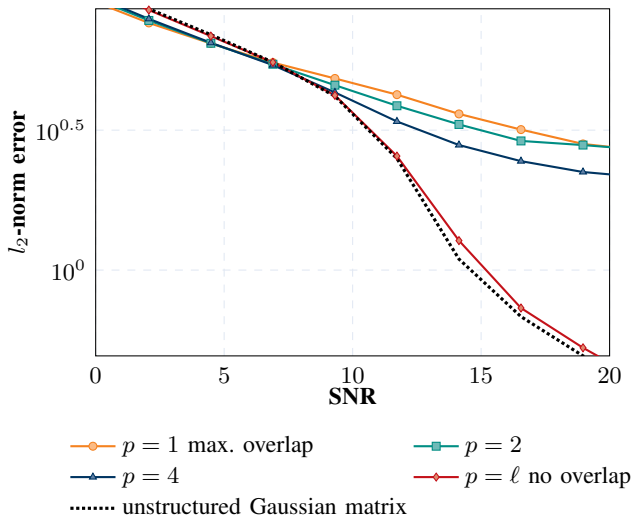
{Influence of overlap $p$ on the reconstruction stage using OMP with $m = 96$ measurements.}{comp_overlap_omp}
\subsection{Vandermonde Matrices}

As we have seen before, in the case of block overlap our method of SOE requires that one of the factors has a Vandermonde structure and motivated by that, we derived Algorithm~\ref{alg:vander_better} in order to get good reconstruction performance even with these rigid structural constraints. To conclude our numerical investigations we present simulation results that display the gain achieved through Algorithm~\ref{alg:vander_better} compared to several other methods for constructing Vandermonde structured sensing matrices and also compared to unstructured matrices drawn from a Gaussian ensemble as a baseline, since they are well known for their good performance during reconstruction~\cite{candes2006nearoptrec}. As a means of reconstruction, we present Orthogonal Matching Pursuit (OMP)~\cite{pati93omp}, because it is a well known and well studied algorithm. Moreover, it is known that the coherence of the involved sensing matrix is involved in performance and stability results of OMP. 

For comparison, we consider Gaussian sensing matrices constructed by drawing each entry identically and independently from a complex zero mean circular symmetric distribution. Then each column is normalized to unit length. We compared these Gaussian matrices with three types of Vandermonde matrices: (a)  independently drawing their generating elements $z_1, \dots, z_N$ from $Y \sim \exp(i 2 \pi U)$, where $U \sim \Unif[0,1]$; (b) deterministically constructing the generating elements by placing them on a regular grid on the complex unit circle, so $z_k = \exp(i 2\pi (k-1)/N)$; (c) according to Algorithm~\ref{alg:vander_better}. As in the case of the Gaussian matrices, we rescale each column of the Vandermonde matrices to unit length.

We also restrict the support sets of the underlying signals $\bm x$ such that the minimum distance between nonzero elements is larger than a certain quantity. To motivate this structural assumption, note that Vandermonde matrices have highly correlated columns, if their generating elements' projections on the complex unit circle are closely located. This has already been established in Theorem \ref{lambda_facts}, where we show that 
\[\lambda(c_1,c_2,\arg(z_1-z_2)) \rightarrow 1 \Text{for} \Abs{\arg(z_1-z_2)} \rightarrow 0.\]
Clearly, if one increases the number of generating elements for the mentioned methods (b) or (c) and fixes one arbitrary generating element, the distance with respect to the adjacent generating elements' arguments decreases and thus the coherence of the resulting $\bm V$ gets closer to $1$. In terms of recovering a sparse vector $\bm x$ in this context means that if two distinct $i,j \in \Supp(\bm x)$ are close to each other in the sense that $\Abs{i -j}/m \leqslant \varepsilon(n)$, then stably recovering this vector becomes impossible. This drawback originates from the structural impositions on $\bm V$ and as such cannot be circumvened. 

In the context where the argument of the generating elements correspond to a parameter $\vartheta \in [0, 2\pi)$ and the columns of $\bm V$ are atoms described by this parameter, e.g. \cite{malioutov2005ssr_source_loc}, we can think of the number of columns $m$ of $\bm V$ as a granularity on $[0,2\pi)$. So increasing $m$ yields a finer grid on the parameter set, or in terms of the function $\lambda$ a finer discrete sampling. However, this does \textit{not} decrease $\varepsilon(n)$ from above and as such does not reduce the distance between two parameters that are present in $\bm x$ and one still can resolve during reconstruction. 

To include these observations in our simulations, we restrict the set of possible supports of the underlying signals $\bm x$. This is done by considering a Vandermonde matrix $\bm V \in \C^{n \times m}$ with its $m$ generating elements regularly placed on the unit circle. Then we just sample support patterns such that the resulting ground truth $\bm x$ only contains columns such that their generating elements' distances with respect to their arguments exceeds $2 \pi / n$. This corresponds to making an assumption about the structure of the occuring signals and is known under the notion structured sparsity \cite{duarte2011structuredCS}.

\linePlot{SNR}{success rate}{
legend style={
legend columns=2,
},
xmin=-15,
xmax=5,
}{
\addplot table [col sep=comma,x=SNR,y=vand_rnd_supp]{matlab/data/vand_96.csv};
\addplot table [col sep=comma,x=SNR,y=vand_reg_supp]{matlab/data/vand_96.csv};
\addplot table [col sep=comma,x=SNR,y=vand_opt_supp]{matlab/data/vand_96.csv};
\addplot table [col sep=comma,x=SNR,y=rnd_supp]{matlab/data/vand_96.csv};
}{
unif. dist. on unit circle, reg. dist. on unit circle , output of Algorithm~\ref{alg:vander_better}, Gaussian matrix}
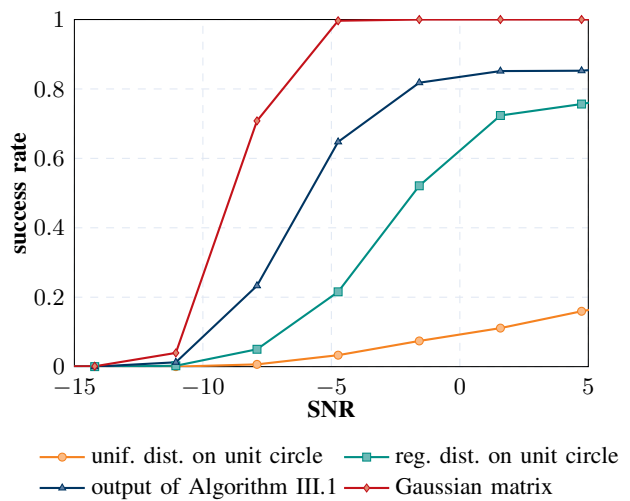
{Comparison of correct support detection during reconstruction using various Vandermonde algorithms against the Gaussian ensemble with $m = 96$ measurements.}{comp_vander_supp}

The results in Figure~\ref{comp_vander_supp} display the effect of Algorithm~\ref{alg:vander_better} on the rate of correct support detection during reconstruction. This is a valid performance measure for reconstruction schemes that produce exactly sparse results, like the OMP algorithm does, since it chooses the amplitudes on the detected support from the solution to the least squares problem on the span of the selected columns. The proposed method yields matrices that improve on the performance of the Vandermonde matrices with elements placed on a uniform grid on the complex unit circle. It is also clear why not all supports are correctly detected, since some of them have non separable elements due to the small distance of the corresponding generating elements of the Vandermonde matrices, which in the presence of noise results in reconstruction errors.

\linePlot{SNR}{MSE of reconstruction}{
legend style={
legend columns=2,
},
xmin=-20,
xmax=20,
ymode=log
}{
\addplot table [col sep=comma,x=SNR,y=vand_rnd_MSE]{matlab/data/vand_96.csv};
\addplot table [col sep=comma,x=SNR,y=vand_reg_MSE]{matlab/data/vand_96.csv};
\addplot table [col sep=comma,x=SNR,y=vand_opt_MSE]{matlab/data/vand_96.csv};
\addplot table [col sep=comma,x=SNR,y=rnd_MSE]{matlab/data/vand_96.csv};
}{
unif. dist. on unit circle, reg. dist. on unit circle , output of Algorithm~\ref{alg:vander_better}, Gaussian matrix}
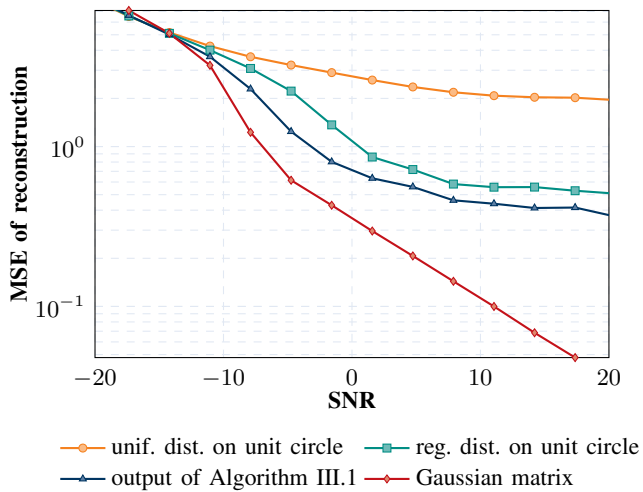
{Comparison in terms of the $l_2$ norm between the ground truth and its reconstruction for various Vandermonde algorithms and the Gaussian ensemble with $m = 96$ measurements.}{comp_vander_mse}

As a second performance metric we also measure the error in the $l_2$-norm between reconstructions and ground truths. These results can be found in Figure~\ref{comp_vander_mse}. Here, the numerical findings correspond to those above and the proposed method displays a lower reconstruction error than the other two means of constructing Vandermonde matrices. As expected, no instance of these highly structured matrices can compete with the performance displayed by matrices from the Gaussian ensemble, which do not have to fulfill any structural requirements.

\section{Conclusions}
In this paper we investigate the problem of estimating the unknown sparsity order from compressive measurements without the need to carry out a sparse recovery stage. We focus on the more challenging case of a single compressed observation vector. As we show, the sparsity order estimation problem can be transformed into a rank estimation problem of a rearranged matrix version of the observation vector if and only if the compressed sensing measurement matrix has a Khatri-Rao structure. Moreover, if the columns of the matrix contain partially overlapping signal blocks, one of the matrices needs to possess a Vandermonde structure. A  larger amount of overlap allows a higher sparsity order to be estimated but also leads to more stringent structural constraints on the measurement matrix.

The Khatri-Rao and Vandermonde constraints on the sompressed sensing measurement motivate us to analyze the suitable choice of the measurement matrices that allows the proposed sparsity order estimation and yet achieves a low coherence. In particular, we analyze the achievable coherence of Vandermonde matrices and propose a design that yields Vandermonde matrices with a low coherence.

Our numerical results demonstrate the trade-off between the maximal sparsity order that can be estimated and the coherence of the corresponding measurement matrices. Moreover, they clearly show the benefit of the proposed low-coherence Vandermonde matrix design.

\bibliographystyle{IEEEsort}
\bibliography{snippets/bib/cs,snippets/bib/sigproc}

\end{document}